\pdfoutput=1
\documentclass[article,11pt,authoryear]{elsarticle}
\usepackage{mydef}
\usepackage{amstext}
\usepackage[hide]{ed}
\usepackage{tikz}
\usetikzlibrary{calc,decorations.markings}
\usetikzlibrary{arrows,patterns,positioning}
\usetikzlibrary{shapes.geometric, backgrounds}

\def\amt{\alpha_1^-}
\def\bpo{\beta_1^+}
\def\bpt{\beta_2^+}
\def\BP{[\amt e^{-X/2} + (\bpo X + \bpt) e^{X/2}]}
\allowdisplaybreaks[3]

\begin{document}

\begin{frontmatter}

\title{Filling the gaps smoothly}

\author[nyu]{A.~Itkin\corref{cor1}}
\ead{aitkin@nyu.edu}
\cortext[cor1]{Corresponding author}

\author[mit,courant]{A. Lipton}
\ead{alexlipt@mit.edu}

\address[nyu]{Tandon School of Engineering, New York University, \\
6 Metro Tech Center, RH 517E, Brooklyn NY 11201, USA}

\address[mit]{Connection Science \& Engineering, Massachusetts Institute of Technology, \\ 77 Massachusetts Avenue, Cambridge MA 02139, USA}
\address[courant]{Courant Institute of Mathematical Sciences, New York University, \\ 251 Mercer Steet, New York NY 10012, USA}

\begin{abstract}
The calibration of a local volatility models to a given set of option prices is a classical problem of mathematical finance. It was considered in multiple papers where various solutions were proposed. In this paper an extension of the approach proposed in \cite{LiptonSepp2011iv} is developed by i) replacing a piecewise constant local variance construction with a piecewise linear one, and ii) allowing non-zero interest rates and dividend yields. Our approach remains analytically tractable; it combines the Laplace transform in time with an analytical solution of the resulting spatial equations in terms of Kummer's degenerate hypergeometric functions.
\end{abstract}


\end{frontmatter}

The local volatility model introduced by \cite{Dupire:94}and \cite{derman/kani:94} is a classical model of mathematical finance. The calibration of the local volatility (LV) surface to the market data, representing either prices of European options or the corresponding implied volatilities for a given set of strikes and maturities, drew a lot of attention over the past two decades. Various approaches to solving this important problem were proposed, see, e.g., \cite {AndreasenHuge2010, LiptonSepp2011iv, ItkinSigmoid2015} and references therein. Below, we refer to \cite{LiptonSepp2011iv} as LS2011 for the sake of brevity~\footnote{We emphasize that the solution proposed in \cite {AndreasenHuge2010} is static in nature, while the solution developed in LS2011 is fully dynamic.}.

There are two main approaches to solving the calibration problem. The first approach attempts to construct a continuous implied volatility (IV) surface matching the market quotes by using either some parametric or non-parametric regression, and then generates the corresponding LV surface via the well-known Dupire formula, see, e.g., \cite{ItkinSigmoid2015} and references therein. To be practically useful, this construction should guarantee no arbitrage for all strikes and maturities, which is a serious challenge for any model based on interpolation. If the no-arbitrage condition is satisfied, then the LV surface can be calculated using \eqref{lv} below, which is equivalent to, but more convenient than, the original Dupire formula. The second approach relies on the direct solution of the Dupire equation using either analytical or numerical methods. The advantage of this approach is that it guarantees no-arbitrage. However, the problem of the direct solution can be ill-posed, \cite{Coleman2001}, and is rather computationally expensive.

An additional difficulty with both approaches is that the calibration algorithm has to be fast in order to be practically useful. On the one hand, analytical or numerical solutions of the Dupire equation are naturally numerically expensive. On the other hand, building a no-arbitrage IV surface could also be surprisingly numerically challenging, because it requires solving a rather involved constrained optimization problem, see \cite{ItkinSigmoid2015}. An additional complication arises from the fact that in the wings the implied variance surface should be linear in the normalized strike, \cite{Lee2004}.

In this paper we extend the approach proposed in LS2011, which is based on the direct solution of the transformed Dupire equation. In LS2011 a piecewise constant LV surface is chosen, and an efficient semi-analytical method for calibrating this surface to the sparse market data is proposed. However, one can argue that ideally the LV function should be continuous in the log-strike space. Below we demonstrate how to extend LS2011 approach by assuming that the local variance is piecewise linear in the log-strike space, so that the corresponding LV surface is continuous in the strike direction (but not in the time direction). While derivatives of the LV function with respect to strike have discontinuities, the option prices, deltas and gammas are continuous. This is to compare with LS2011 where the option prices and deltas are continuous while the option gammas are discontinuous. We also allow for non-zero interest rates and proportional dividends.

The rest of the paper is organized as follows. Section~\ref{DupEq} introduces the Dupire equation and discusses a general approach to constructing the LV surface. Section~\ref{DupireSec} considers all necessary steps for solving the Dupire equation. Section~\ref{secI12} introduces a no-arbitrage interpolation of the source term, which naturally appears when the Laplace transform in time is used, and shows that using this interpolation all the integrals containing this source term can be obtained in a closed form.
Section~\ref{small} considers a special case when the slope of the local variance on some interval is small, so the linear local variance function on this interval becomes flat. Section~\ref{largeA} discusses various asymptotic results which are useful for constructing the general solution of the Dupire equation. Section~\ref{calib} is devoted to the calibration of the model and also describes how to get an educated initial guess for the optimizer. Since computing the inverse Laplace transform could be expensive for small time intervals, Section~\ref{short} describes an asymptotic solution obtained in this limit in \cite{Gatheral2012SmallT} and shows how to use it for our purposes. Section~\ref{ourRes} describes numerical results for a particular set of market data. The final Section concludes. Some additional proofs and derivations are given in two appendices.

\section{Local volatility surface \label{DupEq}}
As a general building block for constructing the local volatility surface we consider Dupire's (forward) equation for the put option price $P$ which is a function of the strike price $K$ and the time to maturity $T$, \cite{Dupire:94}. We assume that the underlying stock process $S_t$ under the risk neutral measure is governed by the following stochastic differential equation
\begin{equation*}
dS_t = (r-q)S_t dt + \sigma(S_t,t) dW_t, \qquad S_0 = S,
\end{equation*}
\noindent where $r \ge 0$ is a constant risk free rate, $q \ge 0$ is a constant continuous dividend yield, $\sigma$ is a given local volatility function, and $W_t$ is the standard Brownian motion. Following \cite{Tysk2012}, we also assume that if $S_t$ can reach 0 in finite time, then 0 is an absorbing barrier. The Dupire equation for the put $P(K,T)$ reads, \cite{Tysk2012}
\begin{align} \label{Dupire}
P_T &= \left\{\dfrac{1}{2}\sigma^2(K,T) K^2 \sop{}{K} - (r-q)K \fp{}{K} - q \right\}P, \nonumber  \\
P(K,0) &= (K - S_0)^+, \quad (K,T) \in (0,\infty) \times [0,\infty), \\
P(0,T) &= 0, \quad P(K,T)_{K \uparrow \infty} = K D, \quad D = e^{-r T}, \nonumber
\end{align}
\noindent where $S_0 = S_t|_{t=0}$.

If the market quotes for $P(K,T)$ are known for all $K,T$, then the LV function $\sigma(K,T)$ can be uniquely determined everywhere by inverting \eqref{Dupire}~\footnote{If the call option market prices are given for some strikes and maturities, we can use call-put parity in order to convert them to put prices, since for calibration we usually use vanilla European option prices.}.
However, in practice, the known set of market quotes is a discrete set of pairs $(K_i,T_j),  \ i=1,\ldots,n_j, j=1,\ldots,M$, where $n_j$ is the total number of known quotes for the maturity $T_j$, which obviously doesn't cover all $K,T$. So the form of $\sigma(K,T)$ remains unknown.

In order to address this issue, it is customary to choose a functional form of $\sigma(K,T)$ for the corresponding time slice. For instance, in LS2011 $\sigma(K,T)$ is assumed to be a piece-wise constant function of $K,T$. The authors propose a general methodology of solving \eqref{Dupire} for their chosen explicit form of $\sigma(K,T)$ by using the Carson-Laplace transform in time and Green's function method in space. This opens the door for using a version of the least-square method for the calibration routine.
Of course, by construction, it makes the whole local volatility surface discontinuous at the boundaries of the tiles, and flat in the wings. While the former feature, in itself, is not necessarily an issue, but should be avoided if possible, the latter feature is somewhat more troubling, since, it is shown in \cite{Friz2013,Friz2015}, that the asymptotic behavior of the local variance is linear in the log strike at both $K \to \infty$ and $K \to 0$. While the result for $K \to 0$ is shown to be true at least for the Heston and Stein-Stein models, the result for $K \to \infty$ directly follows from Lee's moment formula for the implied variance $v_I$, \cite{Lee2004}, and the representation of $\sigma^2$ via the total implied variance $w = v_I T$, \cite{Lipton2001, Gatheral2006}
 \begin{equation} \label{lv}
w_L \equiv \sigma^2(T, K) T = \dfrac{T \dd_T w }{\left(1-\frac{X\dd_X w }{2 w }\right)^2
- \frac{(\dd_X w)^2}{4}\left(\frac{1}{ w }+\frac{1}{4}\right)+\frac{\dd^2_X w }{2}},
\end{equation}
\noindent where $X = \log K/F$ and $F = S e^{(r-q)T}$ is the stock forward price. Therefore, having a flat local volatility deep in the wings should be avoided if possible.

That is why, in this paper, we consider a continuous, piecewise linear local variance $v = \sigma^2(X,T)$ in the spatial variable $X$ for a fixed $T=const$. This allows us to match the asymptotic behavior of $v$ in the wings as well as build the whole surface which is much smoother than in the piecewise constant case. Also, in LS2011 the interest rates and dividends are assumed to be zero, while here we take them into account.

\section{Solution of Dupire's equation \label{DupireSec}}
Introducing a new independent variable $X$ and a new dependent variable $B(X,T) = e^{-X/2} (K D - P(X,T))/Q, \ Q = S e^{-q T}$, which is a scaled covered put, the problem in \eqref{Dupire} can be re-written as follows
\begin{align} \label{Dupire2}
B_T &- \frac{1}{2}v B_{XX} + \dfrac{1}{8}v B = 0,  \\
B(X,0) &= \frac{K - (K-S)^+}{S} e^{-X/2} = e^{- X/2} {\mathbf 1}_{X > 0} + e^{X/2} {\mathbf 1}_{X \le 0}, \nonumber \\
B(X,T)_{X \uparrow -\infty} &= 0, \quad B(X,T)_{X \uparrow \infty} = 0,
\quad (X,T) \in (-\infty,\infty) \times [0,\infty). \nonumber
\end{align}
A similar transformation is used in \cite{Lipton2002} in order to solve the backward Black-Scholes equation. Suppose that there are option price quotes (at least for one strike) for $M$ different maturities $T_1,\ldots,T_M$. Also suppose that for each $T_j$ the market quotes are provided at $X_i, \ i=1,\ldots,n_j$. Then the corresponding continuous piecewise linear local variance function $v_j(X)$ on the interval $[X_{i},X_{i+1}]$ reads
\begin{equation} \label{vDef}
v_{j,i}(X) = v^0_{j,i} + v^1_{j,i} X.
\end{equation}
Subindex $i=0$ in $v^0_{j,0}, v^1_{j,0}$ corresponds to the interval $(-\infty, X_1]$.
Since $v_j(X)$ is continuous, we have
\begin{equation} \label{cont}
v^0_{j,i} + v^1_{j,i} X_{i+1} = v^0_{j,i+1} + v^1_{j,i+1} X_{i+1}, \quad i=0,\ldots,n_j-1.
\end{equation}
The first derivative of $v_j(X)$ experiences a jump at the points $X_{i}, \irg{i}{1}{n_j}$.

Further, assume that $v(X,T)$ is a piecewise constant function of time, i.e.
$v^0_{j,i}, v^1_{j,i}$ don't depend on $T$ on the intervals $[T_j, T_{j+1}], \ j \in [0,M-1]$, and jump to new values at the points $T_j, \irg{j}{1}{M}$. In the original independent variables $K,T$ this condition implies that
\begin{equation*}
v(K_i,T) \equiv v_{j,i} = v^0_{j,i} + v^1_{j,i} \left[\log (K_i/S) - (r-q)T\right], \quad T \in [T_j,T_{j+1}],
\end{equation*}
\noindent i.e. that the local variance is a (discontinuous) piecewise linear function of time $T$. In other words, in the original log-variables $(\log K,T)$ the function $v(\log K,T)$ is piecewise linear in both variables, while in the transformed variables $(X,T)$ the function $v(X,T)$ is piecewise linear in $X$ and piecewise constant in $T$. Thus, $X$ can be viewed as an automodel variable~\footnote{This terminology is borrowed from aerodynamics and physics of gases and fluids.}.

With the above assumptions in mind, \eqref{Dupire2} can be solved by induction. One starts with $T_0 = 0$, and on each time interval $[T_{j-1},T_j], \ \irg{j}{1}{m}$ solves the modified problem for $B_j(X,\tau)$
\begin{align} \label{Dupire3}
&B_{j,\tau} - \frac{1}{2} v_j(X) B_{j,XX} + \dfrac{1}{8}v_j(X) B_j = 0, \\
&B_1(X,0) = B(X,0),  \quad B_j(X,0) = B_{j-1}(X,\tau_{j-1}), \ j > 1 \nonumber \\
&B(X,\tau)_{X \uparrow \pm \infty} = 0, \quad (X,\tau) \in (-\infty,\infty) \times [0,\tau_j], \nonumber
\end{align}
\noindent where $\tau = T - T_{j-1}, \tau_j \equiv T_j - T_{j-1}$, and $B_j$ is the solution of \eqref{Dupire2} corresponding to the time interval $T_{j-1} \le T \le T_j, \irg{j}{1}{m} $.

To solve \eqref{Dupire3}, similarly to LS2011, we use the Carson-Laplace transform $\B = \mathcal{L}(p)\{B\}$ of \eqref{Dupire3} (for application of the Laplace transform to derivatives pricing, see \cite{Lipton2001}) to obtain
\begin{align} \label{Laplace}
- \frac{1}{2} v_j(X) \B_{j,XX} + \left(\dfrac{v_j(X)}{8} + p\right)\B_j &= p B_{j-1}(X,\tau_{j-1}), \\
\B(X,p)_{X \uparrow \pm \infty} &= 0. \nonumber
\end{align}
Since $v(X)$ is a piecewise linear function, the solution of \eqref{Laplace} can also be constructed separately for each interval $[X_{i-1},X_i]$. By taking into account the explicit representation of $v(X)$ in \eqref{vDef}, from \eqref{Laplace} for the $i$-th spatial interval we obtain
\begin{align} \label{Laplace2}
(b_2 &+ a_2 X) \B_{j,XX} + (b_0 + a_0 X)\B_j = p B_{j-1}(X,\tau_{j-1}), \\
b_2 &= - v^0_{j,i}/2, \ a_2 = -v^1_{j,i}/2, \ b_0 = p + v^0_{j,i}/8, \ a_0 = v^1_{j,i}/8. \nonumber
\end{align}
\eqref{Laplace2} is an {\it inhomogeneous} Laplace equation, \cite{PolyaninSaitsevODE2003}. It is well known that if $y_1=y_1(X)$, $y_2=y_2(X)$ are two fundamental solutions of the corresponding {\it homogeneous} equation, then the general solution of \eqref{Laplace2} can be represented as
\begin{align} \label{solInhom}
\B &= C_1 y_1 + C_2 y_2 + p I_{12} \\
I_{12} &= y_2 \int \dfrac{ y_1 B_{j-1}(X,\tau_{j-1})}{(b_2 + a_2 X)W}d X - y_1 \int  \dfrac{ y_2 B_{j-1}(X,\tau_{j-1})}{(b_2 + a_2 X)W}d X, \nonumber
\end{align}
\noindent where $W = y_1 (y_2)_X - y_2 (y_1)_X$ is the so-called Wronskian corresponding to the chosen solutions $y_1,y_2$. Thus, the problem is reduced to finding suitable fundamental solutions of the homogeneous Laplace equations. Based on \cite{PolyaninSaitsevODE2003}, if $a_2 \ne 0$ and $a_0 \ne 0$, the general solution reads
\begin{align} \label{homog}
\B_j &= e^{k X} \mathcal{J}\left(a,0,2 k (\mu - X)\right), \\
k &= \sqrt{-a_0/a_2} = \pm \dfrac{1}{2}, \quad \mu = -\dfrac{b_2}{a_2} = - \dfrac{v^0_{j,i}}{v^1_{j,i}}, \quad a = \dfrac{b_2 k^2 + b_0}{2 a_2 k}. \nonumber
\end{align}
Here $\mathcal{J}(a,b,z)$ is an arbitrary solution of the degenerate hypergeometric equation, i.e., Kummer''s function, \cite{as64}. Two types of Kummer''s functions are known, namely $M(a,b,z)$ and $U(a,b,z)$, which are Kummer''s functions of the first and second
kind~\footnote{Due to the linearity of the degenerate hypergeometric equation any linear combination of Kummer''s functions also solves this equation.}.

\subsection{Numerically satisfactory solutions \label{constr1}}

On each interval of interest, we need to use a fundamental pair that is numerically satisfactory, \cite{Olver1997}.  Since our boundary conditions are set at positive and negative infinity, we need a numerically satisfactory solution for the whole real line. However, it is well known that a single pair of Kummer's functions cannot be numerically satisfactory throughout the whole real line. To overcome this problem, a combined solution can be constructed; below we describe our construction in some detail.

As a preliminary notice, observe that based on the definitions in \eqref{homog}, \eqref{Laplace2} the variable $z$ can be re-written as $z = -2 k v_{j,i}/v^{1}_{j,i}$. Based on the usual shape of the local variance curve and its positivity, see, e.g., \cite{ItkinSigmoid2015} and references therein, for $X \to -\infty$, we expect that $v^1_{j,i} < 0$.  Similarly, for $X \to \infty$ we expect that $v^1_{j,i} > 0$. In between these two infinite limits the local variance curve for a given maturity $T_j$ is assumed to be continuous, but the slope of the curve could be both positive and negative. Also $v_{j,i} \ge 0 \ \  \forall X \in \mathbb{R}$, and $a = -p/(k v^1_{j,i})$. With these observations in mind, we now present our methodology~\footnote{The case where the local volatility is flat on some interval, i.e. $a_2 = 0$, and $a \to \infty$, is considered in Section~\ref{small}.}.

\subsubsection{$v^1_{j,i} < 0$ \label{vNeg}}

For every interval where $v^1_{j,i} < 0$, i.e. $\forall \irg{i}{1}{n_j}$ such that $X \in [X_{i-1}, X_i], \ X_0 = -\infty, \ X_i \le X_{m_j}$, as the first independent solution of the Kummer equation we take $Y_1(z) = z U(a+1,2,z)$ with $k=1/2$, which means
$a > 0$~\footnote{Since in our case $b=0$, by Kummer's transformation, \cite{Olver1997, as64}, $U(a,0,z) = z U(a+1,2,z)$.}.  From the definition of $z$ it follows that $X = \mu - z/(2k) = \mu - z$. Thus, when $X \to -\infty$ we have $z \to \infty$ and $e^{k X} Y_1(z) \to 0$.

This solution is numerically stable across the whole interval $X < X_{m_j}$ except the point $z=0$, which corresponds to $X = \mu$, or, equivalently, $v_{j,i} = 0$; this point belongs to our interval if $\mu < X_{m_j}$~\footnote{As the local variance is linear and non-negative, either this point is at the edge of the interval, or the local variance is flat and vanishes on this interval.}.
At $z=0$ the solution has a branch point, \cite{Olver1997}. The principal branch of $U(a,b,z)$ corresponds to the principal value of $z^{-a}$ and has a cut in the $z$-plane along the interval $(-\infty,0]$.
However, one can observe, that at $z=0$ \eqref{Laplace2} becomes a degenerated ODE, and its solution immediately reads
\begin{equation} \label{degSol}
\B_j = p \dfrac{B_{j-1}(X,\tau_{j-1})}{b_0 + a_0 X} = B_{j-1}(X,\tau_{j-1}).
\end{equation}
Therefore, we can exclude this special case from the below consideration, while if this case were to occur during the actual calibration, we just use the special solution given by \eqref{degSol} instead of the general solution.

As the second independent solution of Kummer's equation for $v^1_{j,i} < 0$ (or $X < X_{m_j}$) we have two choices:
$Y_2(z) = z e^z U(1-a,2,-z)$ or $Y_2(z) = z M(a+1,2,z)$.
It can be shown that if we take the former with $k=-1/2$ (so $a < 0$ and $X = \mu + z$), then two solutions $e^{-X/2} Y_2(X)$ and $e^{X/2} Y_1(X)$ differ just by a constant $e^{\mu}$, so that they are not independent. Therefore, we are compelled to keep $k=1/2$, $a > 0$, and $X = \mu - z$. However, then the function $z e^{z + k X} U(1-a,2,-z)$ diverges at both $X \to -\infty$ and at $z \to 0$.
Similarly, the function $z M(a+1,2,z)$ also diverges at $X \to -\infty$, but is numerically stable at $z \to 0$.

Thus, we have to put $C_2 = 0$ in \eqref{solInhom} on the very first interval $(-\infty,X_1]$ to preserve the boundary conditions at $X \to -\infty$. However, the solution $Y_2(z)$ still contributes to $I_{12}$. In what follows we will use $Y_2(z) = z M(1+a,2,z)$, and show that $I_{12}$ converges in the limit $X \to -\infty$.

For future reference, note that the solutions $y_1(z) = e^{k X} z M(a + 1, 2, z)$ and $y_2 = e^{k X} z U(a + 1, 2, z)$ can also be re-written in terms of Whittaker's functions $M_{p,s}(z), W_{p,s}(z)$, \cite{as64}
\begin{equation*}
y_1(z) = e^{k\mu} M_{-a,1/2}(z), \qquad y_2(z) = e^{k\mu} W_{-a,1/2}(z).
\end{equation*}

\subsubsection{$v^1_{j,i} > 0$ \label{vPos}}

For every interval where $v^1_{j,i} > 0$, i.e. $\forall \irg{i}{2}{n_j+1}$ such that $X \in [X_{i-1}, X_i], \ X_{n_j+1} = \infty, \ X_i > X_{m_j}$, as the first independent solution of Kummer''s equation we again take $Y_1(z) = z U(a+1, 2,z)$ with $k=-1/2$, which means $a > 0$, and $X = \mu + z$. Thus, when $X \to \infty$ we have $z \to \infty$ and $e^{k X} Y_1(z) \to 0$.

Again, this solution is numerically stable on the whole interval $X > X_{m_j}$ except for a singularity at $z=0$ (if $\mu > X_{m_j}$). However, the solution at $z=0$ of \eqref{Laplace2} was already given in the previous subsection.

As far as the second numerically stable solution is concerned, the analysis of the previous subsection is applicable here as well. Therefore, we also take $Y_2(z) = z M(1+a,2,z)$ again with $k=-1/2$, so $a > 0$, and $X = \mu + z$. Accordingly, in \eqref{solInhom} we put $C_2 = 0$
on the very last interval $[X_{n_j},\infty)$ to preserve the boundary conditions at $X \to \infty$.

\subsection{The combined solution across the whole real line \label{wholeSol}}

The solutions described in the previous section are schematically represented in
Table~\ref{TabWhole}.
\begin{table}[!htb]
\begin{center}
\footnotesize
\begin{tabular}{|c|c|c|c|c|c|c|c|}
\hline
Interval & $v^1_{j,i}$ & $k$ & $z$ & $y_1$ & $y_2$ & $C_2$ \\
\hline
$(-\infty, X_1]$ & $< 0$ & 1/2 & $\mu - X$ &   $e^{X/2} z U(a+1,2,z)$ &
$e^{X/2} z M(a+1,2,z)$ & 0 \\
\hline
$[X_i, X_{i+1}]$ & $< 0$ & 1/2 & $\mu - X$ &   $e^{X/2} z U(a+1,2,z)$ &
$e^{X/2} z M(a+1,2,z)$ & \mbox{fc}  \\
\hline
$[X_i, X_{i+1}]$ & $> 0$ & -1/2 & $ X-\mu$ &   $e^{-X/2} z U(a+1,2,z)$ &
$e^{-X/2} z M(a+1,2,z)$ & \mbox{fc}  \\
\hline
$[X_{m_j}, \infty)$ & $> 0$ & -1/2 & $ X-\mu$ &   $e^{-X/2} z U(a+1,2,z)$ &
$e^{-X/2} z M(a+1,2,z)$ & 0  \\
\hline
\end{tabular}
\caption{Our construction of numerically satisfactory Kummer''s pairs. Here "fc" means {\it from continuity}.}
\label{TabWhole}
\end{center}
\end{table}
Accordingly, for $j > 1$ the solution in \eqref{solInhom} on the interval $i$ reads
\begin{align} \label{tot1}
\B_i &= C^{(1)}_{1,i} y_{1,i}(z) + C^{(1)}_{2,i} y_{2,i}(z) + p I^{(1)}_{12,i}(X), \\
W &\equiv W_{1,i} =  e^{X} z^2 W[U(1+a_i,2,z), M(1+a_i,2,z)] = -\dfrac{e^{\mu_i} }{\Gamma(a_i+1)}, \nonumber
\end{align}
\noindent where $\Gamma(x)$ is the gamma function, $I_{12,i}$ is $I_{12}$ defined in \eqref{solInhom} and computed on the interval $i$, and the superscript $^{(s)}$ in $I_{12,i}^{(s)}$ means that $y_{1,i}(X), y_{2,i}(X)$ (the solutions of the homogeneous equation) in the definition of $I_{12}$ in \eqref{solInhom} are taken on the corresponding area $(s)$~\footnote{We use the notation $C^{(l)}_{1,i}, C^{(l)}_{2,i}$ for the integration constants, where super index $\irg{l}{1}{2}$ marks the corresponding area in Fig.~\ref{Fig3}, and the sub index $i$ marks the interval in the $X$ space.}.

For $j=1$ the term $B_{j-1}(X,\tau_{j-1})$ should be replaced with $e^{X/2}$ if $X \in [X_{i-1},X_i], X  \le 0$, and with $e^{-X/2}$ if $X \in [X_{i-1},X_i], X > 0$ in the definition of $I_{12,i}$.

Also in order to satisfy the boundary conditions, see \eqref{Laplace}, $I_{12}$ in \eqref{solInhom} should vanish when $X \to \pm \infty$. We prove this statement in Appendix~\ref{app1}.

Using these results, we can now proceed to constructing the solution of \eqref{Laplace} on the whole real line $X \in [-\infty,\infty]$ by matching solutions on all the intervals.

Suppose that put prices for $T=T_j$ are known for $n_j$ ordered strikes.
\begin{figure}[ht]
\begin{center}
\begin{tikzpicture}[line/.style={<->},thick, framed]


  \draw[->] (-3.0,0) -- (5,0) node[right] {$X$};
  \draw[->] (0,-0.2) -- (0,3.5) node[right] {$v(X)$};

%

\draw[red,ultra thick] (0.5,1) parabola (4,3);
\draw[red,ultra thick] (0.5,1) parabola (-3,2);
\node at (0,-0.3) {$0$};

\node at (-2,-0.3) {$X_1$};
\node at (-2,1.52) {$\bullet$};
\draw[red, dashed] (-2,0) -- (-2,1.52);
\node at (-2.5,0.5) {$A_1$};

\draw[blue, dashed] (-2,1.52) -- (-3,1.8);

\node at (-1,-0.3) {$X_2$};
\node at (-1,1.2) {$\bullet$};
\draw[red, dashed] (-1,0) -- (-1,1.2);

\node at (-1.5,0.5) {$A_{12}$};
\draw[blue, dashed] (-2,1.52) -- (-1,1.2);

\node at (0.7,-0.3) {$X_{m_j}$};
\draw[red, dotted] (0.5,0) -- (0.5,1);

\node at (1.5,-0.3) {$X_3$};
\node at (1.5,1.15) {$\bullet$};
\draw[red, dashed] (1.5,0) -- (1.5,1.15);

\node at (1,0.5) {$A_{23}$};
\draw[blue, dashed] (-1,1.2) -- (1.5,1.15);

\node at (2.5,-0.3) {$\ldots$};

\node at (3.5,-0.3) {$X_{n_j}$};
\node at (3.5,2.45) {$\bullet$};
\draw[red, dashed] (3.5,0) -- (3.5,2.5);

\draw[blue, dashed] (1.5,1.15) -- (3.5,2.5);
\draw[blue, dashed] (3.5,2.5) -- (4,3.5);
\node at (4,0.5) {$A_{n_j}$};

\node at (2.5,2.2) {$2$};
\node at (2.5,1.3) {$1$};

\end{tikzpicture}
\end{center}
\caption{Construction of the whole solution of the Dupire equation. 1 (red solid line - the real (unknown) local variance curve, 2 (dashed blue line) - a piece-wise linear solution.}
\label{Fig3}
\end{figure}
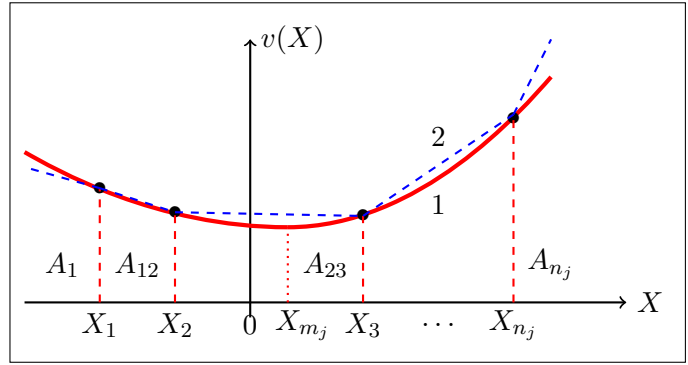

Also, first suppose that these quotes are available for both $K > F$ and $K < F$.  The location of these strikes on the $X$ line is schematically depicted in Fig.~\ref{Fig3}.

According to the analysis of Section~\ref{DupireSec}, on the open interval $A_1$ the solution of \eqref{Laplace} is given by the first line of Table~\ref{TabWhole}. It contains one unknown constant $C^{(1)}_{1,1}$ since we put $C^{(1)}_{2,i} = 0$ due to the boundary conditions. The solutions from line 2 in Table~\ref{TabWhole} should be used for all other intervals $A_{k-1,k}$ such that $k \ge 1, \ v^1_{j,k} \le 0$. These solutions have two yet unknown constants $C^{(1)}_{1,k}, C^{(2)}_{1,k}$, since $X$ is finite on the corresponding interval, and therefore, both solutions $y_1(X), y_2(X)$ are well-behaved. For $X_k$, where $v^1_{j,k} > 0$ and $k \le m_j$ we use the solution given in the third line of Table~\ref{TabWhole}, which also has two yet unknown constants $C^{(2)}_{1,k}, C^{(2)}_{2,k}$ for each interval. Finally, for the interval $X \in [X_{n_j}, \infty) )$, we use the solution in the last line of Table~\ref{TabWhole}.  Again, to obey the boundary conditions we must set $C^{(2)}_{2,n_j+1} = 0$.

Thus, we have $2 n_j$ unknown constants to be determined. Since the local volatility function $v_i$ is continuous at the points $X_i, \ i=1,\ldots,n_j$, so should be the solution $\B(X,\tau_j)$. Therefore, we require that at the points $z_i, \ i = 1,...,n_j$ the solution and its first derivative in $X$ should be a continuous function of $X$. Thus, the above constants solve a system of $2 n_j$ algebraic equations. This system has a special structure that allows one to reduce its LHS matrix to the upper triangular form (actually even the upper banded form). Therefore, it can be efficiently solved with the linear complexity $O(2 n_j)$. For more details, see LS2011.

When computing the first derivatives, we take into account that
\[ h_{i,X} = y_{1,X} I_1(g_i(X))  - y_{2,X} I_2(g_i(X)), \quad i=1,2, \]
\noindent and according to \cite{as64}
\begin{align*}
\fp{M(a,b,z)}{z} &= \dfrac{a}{b}M(a+1,b+1,z), \quad
\fp{U(a,b,z)}{z} = -a U(a+1,b+1,z).
\end{align*}

Also, in some special cases which are discussed in the following sections, the solution can be represented in terms of the modified Bessel functions. But it is known, \cite{as64}, that the derivatives of the modified Bessel functions are expressed in closed form via the same set of functions. Therefore, computing the derivatives of the solution doesn't cause any new technical problems.

Note, that in the definition of the integrals $I_{12}$ in \eqref{solInhom}, for the sake of convenience, we define the low limit of integration $\xi(X)$ as follows. For the interval $A_1$ we take $\xi = -\infty$. Then for each integral $I_{12}(X_i), \ i=2,\ldots,n_j$ we use $\xi_i = X_{i-1}$ (or $z_{i-1}$ if the integral is expressed in $z$ variables, see Appendix). This choice is inspired by the fact that all the parameters $a_i, a_{2,i}, b_{2,i}, \mu_i$ in \eqref{solInhom} are constant on the interval $[X_{i-1}, X_i]$.

Also note, that for the sake of simplicity, in the above construction we assume that market quotes are available for a set of strikes with $K < F$, as well as for a set of strikes with $K > F$. However, it could happen that the market provides just a set of strikes such that all $X_i > 0$ or $X_i < 0$. In this case we can construct the whole solution as follows.

Suppose $X_i < X_{m_j}, \ \forall \irg{i}{1}{n_j}$. Introduce an additional auxiliary point $X_* > X_{m_j}$. Of course, since this is an auxiliary point, the corresponding market quote is not available. Therefore, we don't need to calibrate the local variance at this point. However, introduction of such a point helps to construct the solution on the whole real line, similarly to how it was done above. An unknown constant $C^{(2)}_{1,*}$ again can be found assuming the continuity of the solution at the point $X_*$, while $C^{(2)}_{2,*}$ should be set to 0 to preserve the boundary conditions. Thus, this trick just helps to construct a numerically stable solution across the region $(-\infty, X_1]$ with $X_1 > X_{m_j}$ (when there are no points $X_i < X_{m_j}$), or across the region $[X_{n_j},\infty)$ with $X_{n_j} < X_{m_j}$ (when there are no points $X_i > X_{m_j}$).

According to our construction, the options values as well as option deltas and gammas are continuous in $X$ (and, therefore, in $S$).  Indeed, in the above we required $\B(X,\tau_j), \B_X(X,\tau_j)$ and $v_j$ to be continuous functions of $X$. Then, based on \eqref{Laplace}, $\B_{XX}$ is also a continuous function of $X$. Applying the inverse Laplace transform, we obtain that $B_{XX}$ is also continuous in $X$, and, hence, by the definition of $X$, in $S$. Therefore, by the definition of $B$, $P, \fp{P}{S}, \sop{P}{S}$ are also continuous.
This result demonstrates the additional advantage of our model as compared, e.g., with LS2011, where the options gammas are discontinuous due to only a piece-wise continuity of $v_j$.

\section{Analytical representation of the integrals $I_{12}(X)$ \label{secI12}}

To compute the RHS term $h(X) = p I_{12}(X)$ at some time step $j$ we need a function $B_{j-1}(X,\tau_{j-1})$ obtained at the previous time step. However, market quotes at $T_j$ and $T_{j-1}$ could be given at different sets of $X$ even if the strikes are same since, by definition, $X = \log (K/F(T))$. Therefore, when computing $p I_{12}(X)$ in \eqref{solInhom} by using a numerical quadrature, we need to know the values of $B_{j-1}(X,\tau_{j-1})$ at points $X$ where they have not been calculated yet. There are at least two possible approaches to addressing this issue.

The first approach relies on the fact that the solution $\B_{j-1}$ is already known for each space interval $[X_{i-1},X_i]$. Therefore, to compute $B_{j-1}(X)$, $\ X_{i-1} < X < X_i$ we can use the inverse Laplace transform method as described below. Also this would require computation of $I_{12}(X,\tau_{j-1})$ since this is a part of the solution for $\B_{j-1}$. Thus, this method, despite being exact, is very computationally expensive as it requires the inverse Laplace transform and numerical integration embedded into another inverse Laplace transform and numerical integration.

The second approach, which is advocated in this paper, uses interpolation to compute $B_{j-1}(X)$ given the values of $B_{j-1}(\bar X)$, where $\bar X = X(T_{j-1})$, and $X = X(T_j)$. In general, linear interpolation would be sufficient, however, it gives rise to the violation of the no-arbitrage conditions.

Indeed, according to \cite{CoxRubinstein1985}~\footnote{In \cite{CoxRubinstein1985} these conditions are given for call option prices. In that case the first and the third conditions remain the same as in \eqref{noarb} if we replace $P(K)$ with $C(K)$, while the second condition changes to $C(K_2) > C(K_3).$}, given three put option prices $P(K_1), P(K_2), P(K_3)$ for three strikes $K_1 < K_2 < K_3$, the necessary and sufficient conditions for an arbitrage-free system read
\begin{align} \label{noarb}
P(K_3) &> 0, \qquad P(K_2) < P(K_3), \\
(K_3 - K_2)P(K_1) &- (K_3 - K_1)P(K_2) + (K_2 - K_1)P(K_3) > 0. \nonumber
\end{align}
Suppose that we want to use linear interpolation in the strike space on the interval $[K_1,K_3]$ to find the unknown put option price $P(K_2)$ given the values of $P(K_1), P(K_3)$,
\begin{equation*}
P(K_2) \equiv P_l(K_2) = \dfrac{P(K_1) K_3 - P(K_3) K_1}{K_3 - K_1} + \dfrac{P(K_3) - P(K_1)}{K_3 - K_1} K_2.
\end{equation*}
When plugging this expression into the second line of \eqref{noarb}, the LHS of the latter vanishes, so the third no-arbitrage condition is violated.

This problem, however, could be resolved if we use linear interpolation with a modified independent variable,
\begin{align} \label{lin2}
P(K_2) &\equiv P_F(K_2) \\
&= \dfrac{P(K_1) f(K_3) - P(K_3) f(K_1)}{f(K_3) - f(K_1)} + \dfrac{P(K_3) - P(K_1)}{f(K_3) - f(K_1)} f(K_2), \nonumber
\end{align}
\noindent where $f(K)$ is a convex and increasing function in $[K_1,K_3]$. Indeed, if $f(K)$ is convex, then $P(K_2) = P_F(K_2) = P_l(K_2) - \varepsilon, \ \varepsilon > 0$ (see Fig.~\ref{Fig2}). Substitution of this expression into the second line of \eqref{noarb} gives $(K_3 - K_1) \varepsilon > 0$, which is true. The second condition in \eqref{noarb} now reads
\[ (P(K_1) - P(K_3))(f(K_3) - f(K_2))(f(K_1) - f(K_3)) > 0, \]
\noindent which is also true since $f(K)$ is an increasing function of $K$.

Alternatively, one can use non-linear interpolation. In this paper, for the sake of tractability, we combine both approaches and propose the following interpolation scheme
\begin{align} \label{linNew}
P(K_2) &\equiv P_F(K_2) = \gamma_1 + \gamma_2 K_2 \log K_2, \\
\gamma_1 &= \dfrac{P_3 K_1 \log K_1  - P_1 K_3 \log K_3}{K_1 \log K_1 - K_3 \log K_3}, \nonumber \\
\gamma_2 &= \dfrac{P_1 - P_3}{K_1 \log K_1 - K_3 \log K_3}. \nonumber
\end{align}

\begin{proposition} \label{prop1}
The interpolation scheme in \eqref{linNew} is no-arbitragable.
\end{proposition}
\begin{proof}
Observe, that the no-arbitrage conditions in \eqref{noarb} are discrete versions of the conditions
\[ P > 0, \quad P_K > 0, \quad P_{K,K} > 0. \]
By differentiating the first line of \eqref{linNew} one can check that the proposed interpolation obeys these conditions provided that $P$ is an increasing function of $K$ given the values of all other parameters to be constant. \bs
\end{proof}

\begin{figure}[!htb]
\begin{center}
\fbox{\includegraphics[width=0.98\textwidth]{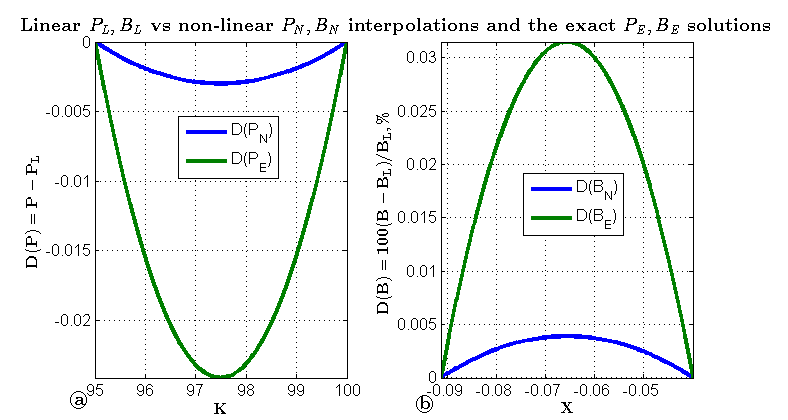}}
\caption{$(a)$: Absolute differences $D(P) = P - P_L$ for no-arbitrage non-linear interpolation $P_N$, and  the exact Black-Scholes put prices $P_E$, with the linear interpolation $P_L$. The line $D(P_L) = 0$ corresponds to $P_L$. \hspace{0.1cm}
$(b)$: Same for the relative differences of no-arbitrage non-linear $B_N$ and linear $B_L$ interpolations with the exact Black-Scholes put prices $B_E$.}
\label{Fig2}
\end{center}
\end{figure}

For the sake of illustration, in Fig.~\ref{Fig2}a we present a comparison of the no-arbitrage interpolation $P_N$ with its linear counterpart $P_L$ and the exact price $P_E$ computed for the Black-Scholes model (for emphasis, the differences $D(P_N) = P_N-P_L, D(P_E) = P_E-P_L$ are displayed). The plot is computed using the following values: $S = 100, K_1 = 95, K_3 = 100, r = 0.05, q = 0.01, \sigma = 0.5, T = 1$. It is clear that the no-arbitrage conditions are satisfied.

Using the definition of $X$ and $B(X,T)$ and some algebra, the interpolation formula in \eqref{lin2} for $B(X)$ can be re-written as
\begin{align} \label{int2}
&B[X,\tau] = \amt e^{-X/2}  + (\bpo X + \bpt) e^{X/2}, \\
\amt &= \dfrac{e^{\frac{X_1+X_3}{2}}}{R} \Big[e^{\frac{X_1+X_3}{2}}(k_3 - k_1) + e^{\frac{X_1}{2}} k_1 B(X_3,T) - e^{\frac{X_3}{2}} k_3 B(X_1,T)\Big], \nonumber \\
\bpo &= \dfrac{1}{R} \left[
e^{X_3} - e^{X_1} - e^{\frac{X_3}{2}}  B(X_3,T) + e^{\frac{X_1}{2}}  B(X_1,T) \right], \nonumber \\
\bpt &= \dfrac{1}{R} \left[ \left(e^{\frac{X_1}{2}} B(X_1,T)- e^{\frac{X_3}{2}} B(X_3,T) \right)\log F + e^{X_1} X_1 - e^{X_3} X_3 \right], \nonumber \\
R &= k_1 e^{X_1} - k_3 e^{X_3}. \nonumber
\end{align}
\noindent where $k_i = X_i + \log F$. In Fig.~\ref{Fig2}b the relative difference for linear $B_L$ and non-linear $B_N$ interpolations vs. the exact Black-Scholes value $B_E$ is shown as a function of $X$. It can be seen that in this test the difference is around 3 bps.

Now the expression given by \eqref{int2} can be substituted into the definition of $I_{12}(X)$ in \eqref{solInhom}. It turns out that then the corresponding integral can be computed in closed form, see Appendix~\ref{app2}.

Sometimes, it could happen that for the new maturity deep out-of-the-money (OTM) or in-the-money (ITM) strikes are positioned outside of the region covered by strikes given for the previous maturity. That means, that no-arbitrage interpolation cannot be used in such a case, while using \eqref{int2} for extrapolation will lead to arbitrage. This issue can be addressed as follows.

Suppose that at $T_j$ the last strike with a known market quote is $K_{j,n_j}$. Suppose that at $T_{j+1} > T_j$ we are given a set of new strikes $K_{j+1,1},...,K_{j+1,n_{j+1}}$, such that $X_{j+1,l} > X_{j,n_j}, \ \forall l: n_{j+1} \le  l \le i$, where $i$ is some integer
$\irg{i}{1}{n_{j+1}}$. Now introduce an auxiliary point $X_{j,*}$, such that $X_{j,*} > X_{j+1,n_j+1}$ and $X_{j,*} < \infty$. Based on the boundary conditions we can assume that $B(X_*,T_j) = 0$. Then, having this extra auxiliary point, the problem of extrapolation reduces to interpolation which was discussed above. A similar approach can be used at the opposite end when $X_{j+1,l} < X_{j,1}, \ \forall \irg{l}{1}{i}$ for some $i > 0$. Then the auxiliary point $X_{j,*}$ should be inserted on the interval $-\infty < X_{j,*} < X_{j+1,1}$.

\paragraph{\bf Solution for the first term $T_1$} For the first term $T_1$ we don't need interpolation since we know the solution $B(X,0)$ along the whole real line $X \in (-\infty,\infty)$. It is given by the terminal condition in \eqref{Dupire2} and fits into our interpolation formula in \eqref{int2} if we set $\amt = p{\mathbf 1}_{X > 0}$ and $\bpo = 0, \bpt = p{\mathbf 1}_{X \le 0}$. Thus, in this case the analytical solution for $I_{12}(X)$ is still available, see Appendix~\ref{app2}.

\section{Small $|v^1_{j,i}|$ \label{small}}

When calibrating the model to the market data, it could happen that some values of $v^1_{j,i}$ become small, so that $|v^1_{j,i} X_i| \ll 1$~\footnote{The case $v^1_{j,i} = 0$ is considered in LS2011, where the integrals $I_{12}(X)$ were computed numerically.}. In this case, the solutions considered in Section~\ref{DupireSec} are no longer valid. Therefore, we need to consider \eqref{Laplace2} which can be represented in the form
\begin{equation} \label{LapPWconst}
\left(1 + \epsilon \right)\B_{j,XX} + \left(\kappa^2 - \dfrac{\epsilon}{4}\right)
\B_j = \dfrac{p}{b_2} B_{j-1}(X,\tau_{j-1}),
\end{equation}
\noindent where $\kappa = \sqrt{b_0/b_2}$, and for each interval $[X_{i-1},X_i],
\irg{i}{2}{n_j}$ the parameter $\epsilon$ is defined as
\begin{equation*} \label{eps}
\epsilon = v^1_{j,i} X_i/v^0_{j,i}.
\end{equation*}
If $|\epsilon_i| \ll 1$, a general solution of \eqref{LapPWconst} $\B_{j}$ can be represented as a power series in $\epsilon$, i.e.,
\[ \B_{j} = \sum_{s=0}^\infty \B^{(s)}_{j}(X) \epsilon^s. \]

\paragraph{\bf Zeroth-order approximation} In zeroth-order approximation \eqref{LapPWconst} can be written as
\begin{equation*}
\B^{(0)}_{j,XX} + \kappa^2 \B^{(0)}_j = \dfrac{p}{b_2} B_{j-1}(X,\tau_{j-1}),
\end{equation*}
So that the corresponding variance is piece-wise constant. A general solution of this equation has the form
\begin{align} \label{LapPWconst2}
\B^{(0)}_{j} &= C_1 y_1(X) + C_2 y_2(X) + \dfrac{p}{b_2} I_{12}(X), \\
y_1 &= e^{\imath \kappa X}, \quad y_2 = e^{-\imath \kappa X}, \nonumber \\
I_{12} &= y_2 \int \dfrac{y_1 B_{j-1}(X,\tau_{j-1})}{W}d X - y_1 \int  \dfrac{y_2 B_{j-1}(X,\tau_{j-1})}{W} d X. \nonumber
\end{align}
Obviously, for these $y_1, y_2$ (which are always numerically satisfactory), we have $W[y_1,y_2] = - 2 \imath \kappa$. Again, we use the no-arbitrage interpolation of the solution obtained at the previous time step to compute $I_{12}(X)$ explicitly:
\begin{align*}
I_1&(X, \kappa) = \int e^{\imath \kappa X} B_{j-1}(X,\tau_{j-1})\dfrac{d X}{W} \\
&= -\dfrac{1}{2\imath \kappa} \int e^{\imath\kappa X} [\amt e^{-X/2}
+ (\bpo X + \bpt) e^{X/2}] d X \nonumber \\
&= -\dfrac{1}{2 \imath \kappa} \left[\dfrac{\amt}{\delta_-} e^{\delta_- X} +
\dfrac{\bpo(\delta_+ X -1) + \delta_+ \bpt}{\delta^2_+} e^{\delta_+ X} \right], \quad \delta_\pm =  \imath \kappa \pm 1/2, \nonumber\\
I_2&(X) = -\dfrac{1}{2\imath \kappa} \int e^{-\imath \kappa X} [\amt e^{-X/2}
+ (\bpo X + \bpt) e^{X/2}] d X \nonumber \\
&= \dfrac{1}{2 \imath \kappa} \left[\dfrac{\amt}{\delta_+} e^{-\delta_+ X} +
\dfrac{\bpo(\delta_- X + 1) + \delta_- \bpt}{\delta^2_-} e^{-\delta_- X} \right], \nonumber \\
I_{12}&(X) = e^{- \imath\kappa X} I_1(X, \kappa) - e^{\imath \kappa X} I_2(X,\kappa)
=  -\dfrac{1}{\imath \kappa}\left[A_- e^{-X/2} + A_+ e^{X/2}\right], \nonumber \\
&A_- = \amt \left(\dfrac{1}{\delta_+} + \dfrac{1}{\delta_-}\right), \nonumber \\
A_+ &= \dfrac{\bpo (\delta_- X + 1) + \delta_- \bpt}{\delta_-^2} +
\dfrac{\bpo (\delta_+ X - 1) + \delta_+ \bpt}{\delta_+^2}. \nonumber
\end{align*}
These solutions can be considered as a further improvement of LS2011, since i) they embed a no-arbitrage interpolation, and ii) this interpolation allows computation of the source terms in closed form. Obviously, performance-wise such an approach significantly speeds up the calculations.

\paragraph{\bf First-order approximation} In first-order approximation in $\epsilon \ll  1$, \eqref{LapPWconst} has the form
\begin{align*}
X \B^{(1)}_{j,XX} &+ (2 + \kappa^2 X) \B^{(1)}_j = X \mathcal{B}^{(0)}(X), \\
\mathcal{B}^{(0)}(X) &= \left( \kappa^2 + \frac{1}{4}\right) \B^{(0)}_j - \dfrac{p}{b_2} B_{j-1}(X,\tau_{j-1}). \nonumber
\end{align*}
If $|X| \ll 1$, then
\[ \B^{(1)}_j = \dfrac{X}{2}\mathcal{B}^{(0)}(X). \]
Otherwise, the solution to this equation reads, see \cite{PolyaninSaitsevODE2003},
\begin{align} \label{LapPWconst2}
\B^{(1)}_{j} &= C_2 + C_1 y_1(X) + I_{12}(X), \\
y_1 &= -\kappa^2 \mathrm{Ei}(-\kappa^2 X) - \dfrac{e^{-\kappa^2 X}}{X}, \qquad
W = \dfrac{e^{-\kappa^2 X}}{X^2}, \nonumber \\
I_{12} &= \int \dfrac{y_1 X \mathcal{B}^0(X)}{W}d X - y_1 \int  \dfrac{X \mathcal{B}^0(X)}{W}d X, \nonumber
\end{align}
\noindent where $\mathrm{Ei(X)}$ is the exponential integral, \cite{as64}. If $\kappa^2 > 0$ then $C_1$ should be set to zero when $X \to -\infty$, i.e., on the very first interval. If $\kappa^2 < 0$ $C_1$ should be set to zero when $ X \to \infty$, i.e., on the very last interval.

\section{Large values of the parameter $|a|$. \label{largeA}}

In many practical situations the parameter $|a|$ in \eqref{homog} can become large. Indeed, it follows from the analysis of Section~\ref{wholeSol} that $|a_i| = 2 p/|v^1_{j,i}|$. The values of $p$ we are interested in can be estimated by taking into account the fact that for computation of the inverse Laplace transform we use the Gaver-Stehfest algorithm described in Section~\ref{calib}. Then, by virtue of \eqref{GSA}, $p = s (\log 2)/\tau_j$, where $s$ runs from 1 to $N=12$. Therefore, for a typical value of $\tau_j = 0.1$, $p$ changes in the range from 7 to 83. At the same time, usually $|v^1_{j,i}|=O(0.1)$, so that $|a| \gg 1$.

From \cite{as64, Olver1997} we know that for large values of $a$ the value of $U(a+1,2,z)$ is very small, while the value of $U(1-a,2,-z)$ is very big. Therefore, the computation of unknown constants $C^{(2)}_{1,i}, C^{(2)}_{2,i}$ is difficult, because i) it requires high-precision arithmetic, and ii) it is pretty unstable. On the other hand, in this case we have a small parameter $1/|a| \ll 1$ in \eqref{Laplace2}, so we can find an asymptotic solution of \eqref{Laplace2}.

We start with a rigorous definition of the small parameter $\varepsilon \equiv -k v^1_{j,i}/p$. Here, when choosing the sign of $k$, we should not rely on the analysis of Section~\ref{wholeSol}, because we only require convergence of our asymptotic solution when $\varepsilon \to 0$. Below we assume that $|\varepsilon| \ll 1$~\footnote{In what follows, for simplicity we omit the modulo, i.e. by saying that $\varepsilon$ is small we mean that $|\varepsilon| \ll 1$.}. With this definition in mind, and using definitions in \eqref{homog}, we re-write \eqref{Laplace2} in the form
\begin{equation} \label{LaplAsymp}
\epsilon \bar{X} \B_{j,XX} + \left(2 k - \frac{1}{4}\varepsilon \bar{X}\right) \B_j = 2 k B_{j-1}(X,\tau_{j-1}), \qquad \bar{X} = X - \mu,
\end{equation}
\noindent where $|2 k| = 1$. This equation belongs to the class of singularly perturbed differential equations, \cite{Wasow1987}. It can be solved by using either the method of matching asymptotic expansions, \cite{Nayfeh2000}, or the method of boundary functions, \cite{VasBut1995} which we will use below.

The need for a special method is due to the fact that for a regular asymptotic expansion of the unknown function $\B(X,\tau)$ in a series in powers of the small parameter $\varepsilon$, zeroth-order approximation yields $\B^{(0)}(X,\tau_j) = B_{j-1}(X,\tau_{j-1})$. Here the superscript $^{(0)}$ denotes the order of the approximation. Obviously, this solution, which doesn't does not depend on any free parameter, is incorrect in the vicinity of the end points of the interval $[X_{i-1},X_i]$, where the solution and its first derivative have to be continuous functions of $X$. So we don't have any degrees of freedom to satisfy this continuity. That is why \eqref{LaplAsymp} belongs to the class of singularly perturbed differential equations, which cannot be solved by using regular expansions in powers of $\varepsilon$.

Following \cite{VasBut1995}, we represent the solution of \eqref{LaplAsymp} on the interval $[X_{i-1},X_i]$ in the form
\begin{equation} \label{asymSol}
\B(X,\tau_j) = \sum_{s=0}^\infty \varepsilon^s \B^{*,s}(X,\tau_j) + \sum_{s=0}^\infty \varepsilon^s \Pi^{(s)}(x_{i-1},\tau_j) + \sum_{s=0}^\infty \varepsilon^s \Xi^{(s)}(x_{i},\tau_j).
\end{equation}
Here $\B^{*}(X,\tau_j)$ is the solution of the so-called "reduced" equation, while $\Pi(x_{i-1},\tau_j)$ and $\Xi(x_{i},\tau_j)$ are the so-called boundary functions. The boundary functions vanish far away from the boundaries $X_{i-1}, X_i$ when $\varepsilon \to 0$. On the other hand, they are needed to ensure that the solution satisfiees the boundary conditions. For any small fixed $\varepsilon \ll 1, \ \varepsilon \ne 0$ the asymptotic solution is an approximation of the exact solution which can be obtained up to $O(\varepsilon^N)$ with $N$ being an arbitrary positive integer, see \cite{VasBut1995}.

Also in \eqref{asymSol} $x_{i-1} = (X - X_{i-1})/\sqrt{\varepsilon}$ is the stretched distance to the left boundary, and $x_{i} = (X - X_{i})/\sqrt{\varepsilon}$ is the stretched distance to the right boundary.

Based on the method of \cite{VasBut1995}, in zeroth-order approximation the reduced equation, which follows from  \eqref{LaplAsymp} at $\varepsilon \to 0$, has a trivial solution $\B^{,0*}(X,\tau_j) = B_{j-1}(X,\tau_{j-1})$. Then, from \eqref{LaplAsymp} the boundary function $\Pi^{(0)}(x,\tau_j)$  solves the equation
\begin{equation} \label{Pi}
(x - \mu_i) \Pi^{(0)}_{xx}(x,\tau_j) + 2 k \Pi^{(0)}(x, \tau_j) = 0.
\end{equation}
The latter has the following solution, \cite{PolyaninSaitsevODE2003}
\begin{align} \label{PiSol}
\Pi^{(0)}(x,\tau_j) &= C_1 \phi_{i-1}(x) I_1(2 \phi_{i-1}(x)) + C_2\phi_{i-1}(x) K_1(2 \phi_{i-1}(x)), \\
\phi^2_{i-1}(x) &\equiv - 2 k \left(x - \mu_i\right) = - \dfrac{2 k}{\sqrt{\varepsilon}} \left(X - X_{i-1} - \sqrt{\varepsilon}\mu_i\right). \nonumber
\end{align}
Here $C_1, C_2$ are two integration constants, and $I_1(x), K_1(x)$ are the modified Bessel functions of the first and second kind.

We must prove that $\Pi^{(0)}(x,\tau_j) \to 0$ when $\varepsilon \to 0$. Based on \cite{as64}, we know that this is true for $K_1(2 \phi_{i-1}(x))$ if $k < 0$ since $X > X_{i-1}$, but not for $I_1(2 \phi_{i-1}(x))$. Therefore, in \eqref{PiSol} we must put $C_1 = 0$, and $k = -1/2$. Note, that for $v^1_{j,i} < 0$ $\varepsilon \in \mathbb{C}$, but this is not a problem.

Similar arguments show that for $\Xi^{(0)}(x,\tau_j)$ in zeroth-order approximation in $\varepsilon$ the solution reads
\begin{align*}
\Xi^{(0)}(x,\tau_j) &= C_2\phi_{i}(x) K_1(2 \phi_{i}(x)), \\
\phi^2_{i}(x) & = - \dfrac{1}{\sqrt{\varepsilon}} \left(X - X_{i} - \sqrt{\varepsilon}\mu_i\right). \nonumber
\end{align*}
Thus, finally, zeroth-order asymptotic solution of \eqref{LaplAsymp} has the form
\begin{align*}
\B^{(0)}(X,\tau_j) &= B_{j-1}(X,\tau_{j-1}) + C_1\phi_{i-1}(x) K_1(2 \phi_{i-1}(x)) \\
&+ C_2\phi_{i}(x) K_1(2 \phi_{i}(x)). \nonumber
\end{align*}
The unknown constants $C_1, C_2$ can be found using the method described in the next section. The values of $B_{j-1}(X,\tau_{j-1})$ at the points $X_{i-1},X_i$ can be obtained by using our no-arbitrage interpolation described in Section~\ref{secI12}.

The next approximations in $\varepsilon$ can also be constructed based on the general method of \cite{VasBut1995}. The reduced equation now reads
\[ \bar{X} \B^{(*,0)}_{XX}(X,\tau_j) - \dfrac{1}{4}\bar{X}\B^{(*,0)}(X,\tau_j) + 2 k \B^{(*,1)}(X,\tau_j) = 0, \]
\noindent with the obvious solution
\begin{equation*}
\B^{(*,1)}(X,\tau_j) = \dfrac{1}{2 k} \bar{X}\left[ \dfrac{1}{4}\partial_{X,X}\B_{j-1}(X,\tau_{j-1}) - \B_{j-1}(X,\tau_{j-1}) \right].
\end{equation*}
As $\B_{j-1}(X,\tau_{j-1})$ solves \eqref{Laplace}, we can represent $\partial_{X,X}\B_{j-1}(X,\tau_{j-1})$ with $j > 1$ in the form
\begin{equation*}
\partial_{X,X}\B_{j-1}(X,\tau_{j-1}) = -\dfrac{2 p}{v_{j-1,i}(X)} B_{j-2}(X,\tau_{j-2}) -
\left(\dfrac{2 p}{v_{j-1,i}(X)} + \dfrac{1}{4}\right) \B_{j-1}.
\end{equation*}

The equation for $\Pi^{(1)}(x,\tau_j)$ is
\begin{equation*} \label{var2}
(x - \mu_i) \Pi^{(1)}_{xx}(x,\tau_j) + 2 k \Pi^{(1)}(x, \tau_j) = \dfrac{1}{4}(x- \mu_i)\Pi^{(0)}_{xx}(x,\tau_j) = - \dfrac{k}{2}\Pi^{(0)}(x,\tau_j).
\end{equation*}
This equation is similar to \eqref{Pi}, the only difference being that now it is inhomogeneous. Accordingly, its solution reads
\begin{align} \label{PiSol2}
\Pi^{(1)}(x,\tau_j) &= \phi_{i-1}(x) K_1(2 \phi_{i-1}(x)) + I_{i-1}^{(1)}, \\
I_{i-1}^{(1)} &= -\dfrac{k}{2}\Bigg\{
\phi_{i-1}(x) I_1(2 \phi_{i-1}(x)) \int \dfrac{K_1(2\phi_{i-1}(x))}{\phi_{i-1}(x)}\Pi^{(0)}(x,\tau_j) d x \nonumber \\
&- \phi_{i-1}(x) K_1(2 \phi_{i-1}(x)) \int \dfrac{I_1(2\phi_{i-1}(x))}{\phi_{i-1}(x)}\Pi^{(0)}(x,\tau_j) dx \Bigg\} \nonumber \\
&= -\dfrac{k}{2}\Bigg\{
\phi_{i-1}(x) I_1(2 \phi_{i-1}(x)) \int K^2_1(2\phi_{i-1}(x)) d x \nonumber \\
&- \phi_{i-1}(x) K_1(2 \phi_{i-1}(x)) \int I_1(2\phi_{i-1}(x)) K_1(2\phi_{i-1}(x))dx \Bigg\} \nonumber
\end{align}
From \cite{PrudnikovBrychkov1986} we have
\begin{align*}
\int K^2_1 &(2\phi_{i-1}(x)) d x = \phi_{i-1}(x) K^2_1(2\phi_{i-1}(x)) - K_0(2\phi_{i-1}(x))K_2(2\phi_{i-1}(x)), \nonumber \\
\int I_1 &(2\phi_{i-1}(x)) K_1(2\phi_{i-1}(x)) d x = \int I_1(2\sqrt{x-\mu_i}) K_1(2\sqrt{x-\mu_i}) d x \\
&= 2\int y I_1(2y) K_1(2y) d y = y^2 \Big[ \left(1 + \dfrac{1}{4 y^2}\right) I_1(2y) K_1(2y) \nonumber \\
&- I'_1(2y) K'_1(2y)\Big] - \dfrac{1}{4}, \qquad y \equiv \phi_{i-1}(x), \nonumber \\
I'_1(2y) &= 2 \left[ I_0(2 y) - \dfrac{1}{y}I_1(2y) \right], \qquad
K'_1(2y) = -2 \left[K_0(2 y) + \dfrac{1}{y}K_1(2y)\right]. \nonumber
\end{align*}
We emphasize that in \eqref{PiSol2} we don't need free constants since they already appear in zeroth-order solution. Therefore, the boundary conditions can be satisfied by choosing appropriate values for these constants.

Accordingly, the function $\Xi^{(1)}(x,\tau_j)$ can be found in a similar way. The overall solution is given by the expression \eqref{PiSol2}, where $\phi_{i-1}(x)$ must be replaced with $\phi_{i}(x)$. This finalizes the construction of the first-order approximation.

We will not construct higher order approximations for $\B^{(s)}(X,\tau_j), \ s > 1$ because incorporation of the first two terms already provides a good approximation with the accuracy of $O(\varepsilon^2)$ (since usually $\varepsilon$ is of order 0.1 or less). Also, as we observed in our numerical experiments, using these asymptotic solutions as part of the calibration procedure makes the latter fairly stable.

\section{The calibration procedure} \label{calib}

The calibration procedure runs sequentially for each time step beginning from $j=1$ and up to $j = M$. Given the solution at the previous time step $B_{j-1}(X,\tau_j)$, we proceed by making some initial guess for the parameters $v^0_{j,i}, v^1_{j,i}, \ i=0,\ldots,n_j$~\footnote{If $j=1$ the previous time solution is just the payoff function.}. Actually, we need this guess just for $v^1_{j,i}, \ i=0,\ldots,n_j$ and $v^0_{j,n_j}$, because based on \eqref{cont}
\begin{equation} \label{recurr}
v^0_{j,i} = v^0_{j,n_m} + \sum_{k=i+1}^{n_j} X_{k}(v^1_{j,k} - v^1_{j,k-1}) , \quad i=0,n_j-1.
\end{equation}
So the total number of the unknown parameters to be determined is $n_j + 2$. Since for maturity $T_j$ only $n_j$ market quotes are given, we need two additional conditions to provide a unique solution. For instance, often traders have an intuition about the asymptotic behavior of the volatility surface at infinity, which, according to our construction, is determined by $v^1_{j,n_j}$ and $v^1_{j,0}$.

Using the analytical solution $\B_{j}(X,p)$ for a given maturity, the scaled put option prices $B(X_i,\tau_j)$ can be calculated similarly to LS2011 by computing the inverse Carson-Laplace transform. The latter can be efficiently performed by using the Gaver-Stehfest algorithm
\begin{equation} \label{GSA}
B(X,\tau_j) = \sum_{s=1}^{(N)} \dfrac{St_s^{(N)}}{k} \B(X, s \Lambda), \quad \Lambda = \dfrac{\log 2}{\tau_j}.
\end{equation}
This algorithm was studied in many papers (see, e.g., \cite{Kuznetsov2013} and references therein), and, provided that the resulting function is non-oscillatory, converges very quickly. For instance, choosing $N=12$ is usually sufficient. The coefficients $St_s^{(12)}$ can be found explicitly, see, e.g., LS2011. It is also known that this algorithm requires high-precision arithmetic for its implementation.  This effect is especially pronounced for small $\tau_j$, so the inversion can become numerically unstable unless a sufficient number of significant digits is used.

Once all the option prices are computed, they can be compared with given market quotes. Hence, some kind of a least-square minimization procedure can be utilized to find the final values of all the unknown parameters that fit model option prices to market quotes. Complexity-wise, at every iteration we need to compute the solution at $n_j$ spatial points and $N$ temporal points, the former are given, the latter are prescribed by the Gaver-Stehfest algorithm. Also every such solution, as it is defined in \eqref{tot1}, requires $2 n_j$ constants $C_1,...,C_{2 n_j}$, which solve the corresponding system of linear equations. As it was mentioned earlier, due to its special structure, this system can be solved with complexity of $O(2 n_j)$. Overall, complexity of performing one iteration is $O(2 n_j N)O(\text{Ku})$, where $O(\text{Ku})$ is complexity of computing all Kummer's functions for the solution in one spatial point. This seems to be a significant improvement in performance as compared, e.g., with LS2011, where computation of the source terms required numerical integration.

\subsection{Initial guess for the calibration. \label{initGuess}}

Obviously, the calibration is a time-consuming process, therefore, having a smart initial guess significantly improves its convergence rate.

Suppose we have already obtained all values of the parameters for maturities $T_j, \ j \in [1,j_1], \ j_1 < M$, and now need to run the calibration for the maturity $T_m, \ m = j_1+1$. Also suppose we are given market values $w(m,i), \ i=1,\ldots,n_m$ for the implied variance.  To produce an "educated" initial guess for the calibration procedure, we suggest to use \eqref{lv} to get the initial values of $v^1_{m,i}, \ i=1,\ldots,n_m-1$ and $v^0_{m,n_m}$. In particular, the first derivatives $\dd_T w(m,i), \dd_X w(m,i)$ in the right-hand-side of \eqref{lv} can be approximated by the finite-differences of the first order using two given values of $w$ in the strike and time space, and the second derivative $\dd^2_X w$ - by using the second order approximation with three given values of $w$ in the strike space. When computing $\dd_T w(m,i) \approx [w(m,K_i) - w(m-1,K_i)]/\tau_m$ it is possible that market quote $w(m-1,K_i)$ is not available at $T_{m-1}$; in this case interpolation/extrapolation in $K$ over given quotes at $T_{m-1}$ can be used to get this value. This calculation generate values for $\sigma_{m,i}, \ \irg{i}{1}{n_m}$. If some of them are negative, they can be replaced by a small positive number $\delta$.

Next, we use \eqref{recurr} and obtain a system of linear equations of the form
 \begin{align}
 v^0_{m,n_m} + \sum_{k=i}^{n_m} & X_{k}(v^1_{m,k} - v^1_{m,k-1}) + v^1_{m,i} X_{i} = \sigma_{i,m}, \qquad i \in [1,n_m-1], \nonumber \\
 v^0_{n_m,m} &= \sigma_{n_m,m} - v^1_{n_m,m} X_{n_m}, \nonumber
 \end{align}
\noindent where the values $v^1_{m,n_m}$, $v^1_{m,0}$ are given. Since this system is upper triangular, it could be efficiently solved with linear complexity $O(n_m)$.

\section{Option prices for short $T$ \label{short}}

As was mentioned in the previous section, computation of the inverse Laplace transform by using the Gaver-Stehfest algorithm requires very high-precision arithmetic for small $\tau_j$. Therefore, in this limit it does make sense to solve the modified Dupire's equation in a different way, namely by using an asymptotic expansion for its solution at $\tau_j \to 0$,  see also LS2011.

For the time-homogeneous models of the local volatility, i.e., when the volatility does not depend explicitly on time, this problem was considered in various papers, see, e.g., \cite{Gatheral2012SmallT} and references therein. For the time-inhomogeneous model it was further analyzed in \cite{Gatheral2012SmallT}. In that paper an asymptotic representation for the European call option price $C(\tau_j,x)$ with $x = \log S$ was obtained by using an expansion of the transition density function of a one-dimensional time inhomogeneous diffusion. If $x < \log K$ this asymptotic solution reads
\begin{align} \label{gathT}
C^{-}(\tau_j,x) &= \dfrac{v_j(K) K}{ \sqrt{2\pi} }\dfrac{u_0(x,\log K)}{d^2(x,\log K)}
\tau_j^{3/2} \exp \left[-\frac{d^2(x,\log K)}{2 \tau_j}\right], \nonumber \\
d(x,y) &= \int_x^y \frac{dx}{\sqrt{v_j(x)}}, \\
u_0(x,y) &= \left(\dfrac{v_j(x)}{v^3_j(y)}\right)^{1/4} \exp\left[ - \dfrac{1}{2}(y-x) + (r-q) \int_x^y
\dfrac{d s}{v_j(s)}\right], \nonumber
\end{align}
\noindent where the superscript $^{(-)}$ is used to indicate that this solution corresponds to $x < \log K$~\footnote{In our notation $x = \log K - (r-q)T - X$.}. Also, when deriving \eqref{gathT} it is assumed that $\forall x \in \mathbb{R} \ \ \exists C > 0 : C^{-1} \le \sigma_j(x), |\sigma'_j(x)| \le C, |\sigma''_j(x)| \le C$. This assumption may fail at the boundaries when $S \to 0$ and $S \to \infty$.

Having call prices $C(\tau_j,x_i)$ computed for all strikes $K_i, \ i=1,\ldots,n_j$ and a particular maturity $\tau_j \ll \min_i (1/\sigma_{j,i})$, we can also compute the corresponding put prices by using call-put parity. Then, running parameters for the local variance function can be found by calibration. Note, that since $v_j(x)$ is piece-wise linear in $X$, the integral in \eqref{gathT} can also be constructed as a sum of various contributions. When calculating these contributions, we rely on the fact that if $x,y$ belong to the interval $i$, the variance on this interval is given by \eqref{vDef}, so that
\begin{align*} \label{int}
d(x,y) &= \int_x^y \dfrac{dx}{\sqrt{b - a_2 x}} = \frac{2}{a_2}\left(\sqrt{b_2 + a_2 Y} - \sqrt{b_2 + a_2 X}\right), \\
\int_x^y \dfrac{d s}{v_j(s)} &=
\frac{1}{a_2} \log \frac{b_2 + a_2 Y}{b_2 + a_2 X}, \nonumber
\end{align*}
\noindent where $b = b_2 + a_2 (\log K + \kappa)$.

If $x > \log K$, from \cite{Gatheral2012SmallT} we have
\begin{equation*}
C^{+}(\tau_j,x) = e^x - K e^{-r \tau_j} - C^{-}(\tau_j,x).
\end{equation*}

\section{Results and discussion \label{ourRes}}

\begin{sidewaystable*}[!htb]
\begin{center}
\footnotesize

\caption{XLF implied volatilities for the call options.}
\label{TabOptC}
\begin{tabular}{|c|c|c|c|c|c|c|c|}
\specialrule{.1em}{.05em}{.05em}
T & \multicolumn{7}{|c|}{K, Put} \cr
\cline{2-8} & 18   & 19    & 20    & 21     & 21.5  & 22        & 23    \cr
\specialrule{.1em}{.05em}{.05em}
4/4/2014   &  -    & -     & 39.53 & 23.77  & 19.73 & 16.67 & - \cr
\hline
4/19/2014  &  -    & 32.90 & 26.79 &  20.14 & -     & 15.19 & 12.93 \cr
\hline
5/17/2014  & 33.27 & 26.88 & 23.08 & 18.94  & -     & 16.12 & 13.86  \cr
\hline
6/21/2014  & 27.84 & 23.90 & 21.07 & 18.88  & -     & 16.95 & 15.82   \cr
\hline
7/19/2014  & 26.09 & 22.81 & 20.29 & 18.13  &  -    & 16.30 & 14.93   \cr
\hline
9/20/2014  & 24.20 & 22.23 & 20.32 & 18.76  & -     & 17.40 & 16.41   \cr
\specialrule{.1em}{.05em}{.05em}
\end{tabular}

\bigskip

\caption{XLF implied volatilities for the put options.}
\label{TabOptP}
\begin{tabular}{|c|c|c|c|c|c|c|c|c|c|c|}
\specialrule{.1em}{.05em}{.05em}
T & \multicolumn{10}{|c|}{K, Call} \cr
\cline{2-11} & 21  & 21.5  &  22   &  22.5 & 23    & 24    & 25    & 26    & 27    & 28 \cr
\specialrule{.1em}{.05em}{.05em}
4/4/2014   &  -    & 16.60 & 14.69 & 14.40 & 14.86 & -     &  -    & -     & -     & -  \cr
\hline
4/19/2014  &  -    & -     & 15.79 & -     & 13.38 & 15.39 & -     & -     & -     & - \cr
\hline
5/17/2014  & 16.71 & -     & 14.48 & -     & -     & 13.75 & -     & -     & -     & - \cr
\hline
6/21/2014  & 16.31 & -     & 14.78 & -     & -     & 13.92 & 14.28 & 16.58 & -     & - \cr
\hline
7/19/2014  & 16.82 & -     & 15.24 & -     & -     & 14.36 & 14.19 & 15.20 & -     & -  \cr
\hline
9/20/2014  & 17.02 & -     & 15.84 & -     & -     & 14.99 & 14.56 & 14.47 & 14.97 &16.31 \cr
\specialrule{.1em}{.05em}{.05em}
\end{tabular}

\bigskip

\caption{Typical time to converge (per strike) using various algorithms for computing $\B(X,T)$.}
\label{TabBench}
\begin{tabular}{|c|c|c|c|}
\hline
Method & $T \ll 1$ &  $|a| >> 1$ & general \cr
\hline
Time, sec & 1.0-1.4 & 1-7 & 5-7 \cr
\hline
\end{tabular}

\end{center}
\end{sidewaystable*}

In our numerical test we use the same data set as in \cite{ItkinSigmoid2015}, i.e.,
we take data from \url{http://www.optionseducation.org} on XLF traded at NYSEArca on March 25, 2014. The spot price of the index is $S = 22.64$, and $r = 0.0148,\  q=0.01$. The option implied volatilities (IV) are given in Tables~\ref{TabOptC},\ref{TabOptP}. We take all OTM quotes and some ITM quotes which are very close to the at-the-money (ATM).

When strikes for calls and puts coincide, we take an average of $I_{call}$ and $I_{put}$ with weights proportional to $1 - |\Delta|_c$ and $1-|\Delta|_p$ respecitvely, where $\Delta_c, \Delta_p$ are option call and put deltas~\footnote{By doing so we do take into account effects reported in \cite{callPutIV}, who pointed out that the IVs calculated from  call and  put option prices corresponding to the same strike  do not coincide,  although  they  should  be  equal  in theory. Our weights are chosen according to a pure empirical rule of thumb, and a more detailed investigation of this effect is required.}.

We have already mentioned that in our model for each term the slopes of the smile at plus and minus infinity, $v^1_{j,n_j}$ and $v^1_{j,0}$, are free parameters. So often traders have an intuition about these values. However, in our numerical experiments we took just some plausible values for them, which are given in Table~\ref{TabSlopes}.

\begin{table}[H]
\begin{center}
\begin{tabular}{|c|c|c|c|}
\hline
$j$ & $T_j$ & $v^1_{j,0}$ & $v^1_{j,n_j}$ \cr
\hline
1 & 4/04/2014  & -0.1206 & 0.1000 \cr
\hline
2 & 4/19/2014  &  -0.1000 &  0.1000 \cr
\hline
3 & 5/17/2014  &  -0.1309 &  0.1000 \cr
\hline
4 & 6/21/2014  &  -0.1000 &  0.1000 \cr
\hline
5 & 7/19/2014  &  -0.1000 &  0.1000 \cr
\hline
6 & 9/20/2014  &  -0.1000 &  0.1000 \cr
\hline
\end{tabular}
\caption{Parameters $v^1_{j,0}$ and $v^1_{j,n_j}$ for the option data in Table~\ref{TabOptC}, \ref{TabOptP}.}
\label{TabSlopes}
\end{center}
\end{table}

\begin{figure}[!htb]
\begin{center}
\fbox{\includegraphics[width = 4.5in]{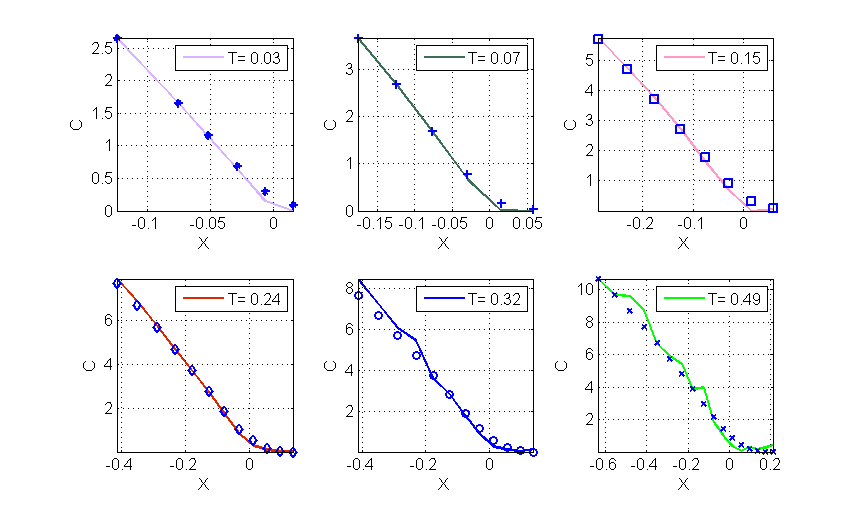}}
\caption{Term-by-term fitting of market prices constructed using the whole set of data in Tab.~\ref{TabOptC},\ref{TabOptP}.}
\label{FigCall}
\end{center}
\end{figure}
When calibrating the model to market data, we use the standard Matlab {\it fmincon} function. We start by using an "active-set" algorithm, and if it doesn't converge, switch to an "sqp" algorithm. We emphasize that optimization of this step is not a subject of this paper, and for a more detailed discussion of various problems related to the calibration of the local volatility surface we refer the reader to a recent paper \cite{Lindholm2014} and references therein. Therefore, here calibration is provided for pure illustrative purposes, and certainly a more sophisticated and powerful algorithm could be used to a greater effect.

The results of such a calibration are given in Fig.~\ref{FigCall}. Here each subplot corresponds to a single maturity $T$ (marked in the legend) and shows market data (discrete points) and computed values (solid line). This simple local calibration algorithm provides rather decent results, except for the vicinity of $X = -0.5$ in the last subplot.

For the first two maturities we successfully use the asymptotic method described in Section~\ref{short}. Then, for the next two maturities, the method described in Section~\ref{largeA} provides good results. Finally, for the last two maturities a combination of the general algorithm with that described in Section~\ref{largeA} has to be used.

The local variance curves obtained as a result of this fitting are given term-by-term in Fig.~\ref{termXLF}. The corresponding local variance surface is represented in Fig.~\ref{lvXLF}
\begin{figure}[!htb]
\begin{center}
\fbox{\includegraphics[width=3.5in]{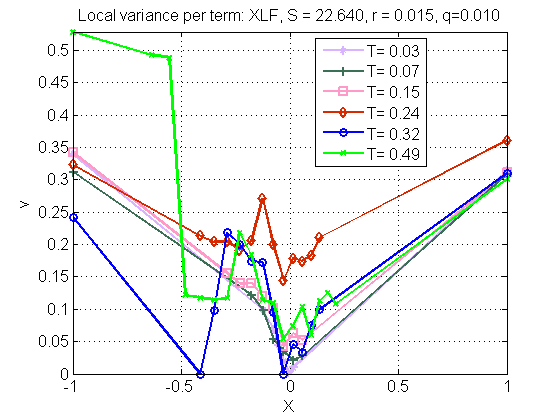}}
\caption{Term-by-term fitting of the local variance.}
\label{termXLF}
\end{center}
\end{figure}

\begin{figure}[!htb]
\begin{center}
\fbox{\includegraphics[width=4in, height=3in]{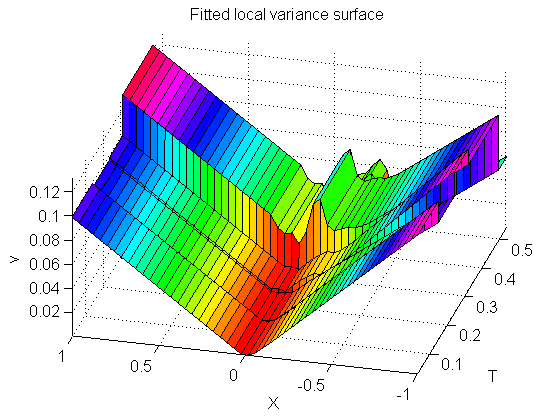}}
\caption{The local variance surface constructed using the proposed approach.}
\label{lvXLF}
\end{center}
\end{figure}

It can be seen that the local variance is positive everywhere on the grid, so that our construction is arbitrage-free.

Performance-wise the proposed algorithm is reasonably efficient. Indeed, we ran our tests in Matlab using two Intel Quad-Core i7-4790 CPUs, each of 3.80 Ghz. As was mentioned in the previous section, the calibration time strongly depends on the method chosen to compute $\B(X,T)$. Typical results are given in Table~\ref{TabBench}. These results are normalized per number of strikes for a given term. Obviously, they could be considered just as a crude estimation, since the convergence strongly depends on the quality of the initial guess. In our calculations we used the approach described in Section~\ref{initGuess}. Still, it can be seen that the second method in Table~\ref{TabBench} is slower than the first one as it requires the evaluation of the Bessel functions. The third method requires multiple computations of Kummer''s functions and is the slowest one. However, as we use the Gaver-Stehfest algorithm, it can be fully parallelized. Same is true for the computation of Kummer''s functions in all points $X_i, \irg{i}{1}{n_j}$ for a given maturity $T_j$, which we do at every iteration of the calibration procedure. Therefore, having a sufficient number of cores, a potential speedup of the parallel implementation should be proportional to $N=12$ (the number of the Gaver-Stehfest algorithm time steps) times the number of strikes. In our case this provides the calibration in less than a second per maturity even when the general method is used.

\section{Conclusion}
In this paper we provide an extension of the approach proposed in LS2011 by replacing a tiled local variance shape with a piece-wise linear construction and relaxing their assumptions about zero interest rates and dividend yields. Yet our approach, which combines an application of the Carson-Laplace transform and solution of the resulting inhomogeneous ordinary differential equation in terms of Kummer''s hypergeometric functions, remains analytically tractable.

When solving the modified Dupire equation by utilizing the Carson-Laplace transform method, one must be cognizant of the following issue. To compute the source term $h(X) = p I_{12}(X)$ at some time step $j$ we need the function $B_{j-1}(X,\tau_{j-1})$ obtained at the previous time step. However, the market quotes for the maturities $T_j$ and $T_{j-1}$ could be given at different sets of $X$ even if the strikes $K$ are same, since, by definition, $X = \log (K/F)$ and $F = F(T)$. Therefore, we need the values of $B_{j-1}(X,\tau_{j-1})$ at certain points $X$ where they have not been calculated yet. LS2011 use interpolation to obtain the required values. However, this interpolation must be carefully constructed to preserve no-arbitrage. In this paper we proposed an interpolation which allows computation of the source terms in closed form, and proved that our interpolation does not create arbitrage.

In addition, we noticed that using the general algorithm for small maturities or steep local variance slopes often results in various inefficiencies and instabilities. Therefore, for these special cases we proposed alternative methods constructed by using asymptotic (regular or singular) expansions, which do not suffer from these issues. In our opinion, this is an interesting and practically important extension of the general methodology described in the previous two paragraphs.

Numerical experiments presented in the paper demonstrate robustness of our approach.
Obviously, closed-form solutions for the source terms and asymptotic solutions, expressed in terms of functions less computationally expensive than Kummer's functions, significantly speed up the calibration. The implementation could be made more efficient by using the internal parallelism of the Gaver-Stehfest algorithm, and the fact that Kummer's functions corresponding to different points $X_i, \irg{i}{1}{n_j}$ for a given maturity $T_j$ could be computed in parallel.

By its nature, our model (as well as any other LV model) provides just a fit for the current market snapshot, and does not consider any dynamics for the local volatility surface {\it itself}. While the latter issue should be investigated separately, our choice of the LV is parsimonious enough to greatly facilitate this endeavor.

\clearpage
\section*{Acknowledgments}
AI is grateful to Christoph Burgard for useful comments.

\section*{References}

\newcommand{\noopsort}[1]{} \newcommand{\printfirst}[2]{#1}
  \newcommand{\singleletter}[1]{#1} \newcommand{\switchargs}[2]{#2#1}

\clearpage

\begin{appendices}
\renewcommand{\theequation}{\Alph{section}.\arabic{equation}}
\setcounter{equation}{0}
\renewcommand{\theequation}{A.\arabic{equation}}

\section{Convergence of $I_{12}(X)$ for $X \to \pm \infty$. \label{app1}}

In this Appendix we prove the following Proposition:
\begin{proposition} \label{Prop1}
For $X \to \pm \infty$ the function $I_{12}(X)$, defined in \eqref{solInhom}, vanishes.
\end{proposition}
\begin{proof}
First, we intend to prove this Proposition for $j=1$. In this case the \eqref{solInhom} has the form
\begin{align} \label{total}
\B &=
\begin{cases}
C_1 y_1 + C_2 y_2 + p h_1(X), & X \le 0, \\
C_3 y_1 + C_4 y_2 + p h_2(X), & X > 0,
\end{cases}
\\
h_i(X) &= y_2 I_1(g_i(X)) - y_1 I_2(g_i(X)), \ i=1,2, \nonumber \\
g_1(X) &= e^{X/2}, \quad g_2(X) = e^{-X/2}, \nonumber \\
I_s(g_l(X)) &= \int_\xi^X y_s \dfrac{g_l(X)}{(b_2 + a_2 X)W}d X,
\quad \irg{s,l}{1}{2}.\nonumber
\end{align}
Thus, in this case $I_{12}(X) = h_1(X)$ if $X \le 0$, and  $I_{12}(X) = h_2(X)$ if $X > 0$. Once this is done, due to the boundary conditions at $X \to -\infty$, the function $\B(X,\tau_j)$ in \eqref{solInhom} tends to $g_1(X)$ in \eqref{total}, and at $X \to \infty$ we have $\B(X,\tau_j) \to g_2(X)$. Therefore, at $X \to -\infty$ we have $I_{12}(X) \to h_1(X)$, and  at $X \to \infty$, similarly  $I_{12}(X) \to h_2(X)$. Thus, the first step of the proof is sufficient to prove the Proposition in its entirety. At $X \to - \infty$ (according to Section~\ref{constr1} this region belongs to the area where $v^1_{j,i} < 0$) we have $z \to \infty$, and, as follows from Table~\ref{TabWhole} and  \eqref{total}
\begin{align} \label{ineq}
I_1(g_1(X)) &=  \int \dfrac{y_1(X) g_1(X)}{(b_2 + a_2 X)W}d X
= \dfrac{\Gamma(a+1)}{a_2} e^{-\mu/2} \int  e^{-z/2} M(1+a,2,z) d z, \\
I_2(g_1(X)) &=  \int \dfrac{y_2(X) g_1(X)}{(b_2 + a_2 X)W}d X = \dfrac{\Gamma(a+1)}{a_2} e^{-\mu/2} \int e^{-z/2} U(a+1,2,z) d z. \nonumber
\end{align}
Thus,
\begin{align} \label{h1}
h_1(z) &= \dfrac{\Gamma(a+1)}{a_2} G(z), \\
G(z) &\equiv e^{-z/2} M(1+a,2,z) \int e^{-z/2} U(a+1,2,z) d z \nonumber \\
&- e^{-z/2} U(1+a,2,z) \int  e^{-z/2} M(1+a,2,z) d z. \nonumber
\end{align}
From \cite{Olver1997}, at $z \to \infty$ we have the following asymptotic series representation
\begin{align} \label{series}
U(a,2,z) &= \Phi_\infty (z), \qquad \Phi_n (z) \equiv z^{-a} \sum_{s=0}^n \dfrac{(a(a-1))_s}{s!}(-z)^{-s}, \\
M(a,2,z) &= \Psi_\infty (z), \qquad \Psi_n (z) \equiv \dfrac{e^{z} z^{a-2}}{\Gamma(a)} \sum_{s=0}^n \dfrac{(1-a)_s (2-a)_s}{s!}z^{-s}, \nonumber
\end{align}
\noindent where $(\cdot)_s$ is the Pochhammer symbol.

Let us define the function $G_n(z)$ in the same way as $G(z)$ in \eqref{h1}, but replacing $U(a,2,z) = \Phi_\infty(z)$ with $\Phi_n(z)$. It is clear that $\lim_{n \to \infty} G_n(z) = G(z)$. Substituting \eqref{series} into this definition and performing integration term-by-term, we arrive at
\begin{align} \label{afterInt}
\int e^{-z/2} \Phi_n(z) &= - 2^{-a}\sum_{s=0}^n f(a,s) 2^{n-s} (-1)^s \Gamma(-s-a,z/2), \\
\int e^{-z/2} \Psi_n(z) &= - (-2)^{a} \sum_{s=0}^n f(a,s) 2^{n-s} (-1)^s \Gamma(-s+a,-z/2), \nonumber \\
f(a,s) &= \dfrac{(a(a-1))_s}{s!}. \nonumber
\end{align}
\noindent where $\Gamma(a,z)$ is an incomplete gamma function. By \cite{Olver1997}, at $z \to \infty$ we have
\[ \Gamma(a,z) = z^{a-1} e^{-z} \sum_{s=0}^\infty (-1)^s \frac{(1-a)_s}{z^s}. \]
Substituting this expression into \eqref{afterInt} and collecting terms, we can check that the leading term in this series is $G_n(z) \sim z^{-2}$. Thus, $G_n \to 0$ at $z \to \infty$ as $1/z^2$. Since this convergence rate doesn't depend on $n$, we can take the limit $n \to \infty$ and see that $G(z) \to 0$ at $z \to  \infty$. Since at $k = 1/2$ we have $z = \mu - X$, that means that that $h_1(X) \to 0$ for $X \to -\infty$.

For $h_2(x)$ the representation for $I_1(g_2(X)), I_2(g_2(X))$ is similar to that in \eqref{ineq} and reads
\begin{align*}
I_1(g_2(X)) &=  \int \dfrac{y_1(X) g_2(X)}{(b_2 + a_2 X)W}d X = -\dfrac{\Gamma(a+1)}{a_2} e^{-\mu/2} \int  e^{-z/2} U(1+a,2,z) d z, \\
I_2(g_2(X)) &=  \int \dfrac{y_2(X) g_2(X)}{(b_2 + a_2 X)W}d X = - \dfrac{\Gamma(a+1)}{a_2} e^{-\mu/2} \int  e^{-z/2} M(1+a,2,z) d z. \nonumber
\end{align*}
Since we need the limit $z \to \infty$, the convergence of these integrals to zero can be proved similarly to the previous case of $z \to -\infty$. Thus, $h_2(X) \to 0$ for $X \to \infty$.
\hfill $\blacksquare$
\end{proof}

\section{Closed form solution for $I_{12}(X)$ \label{app2}}

Here we derive an analytical expression for $I_{12}(X)$ in \eqref{solInhom}, which takes into account our approximation of $B(X,\tau_{j-1})$ presented in Section~\ref{secI12}, and reads
\begin{align*}
I_{12}(X) &= y_2 I_1(X) - y_1 I_2(X), \\
I_1(X) &= \int_\xi^X \dfrac{y_1 B_{j-1}(X,\tau_{j-1})}{(b_2 + a_2 X)W}d X \nonumber \\
&= \int_\xi^X \dfrac{y_1 \BP }{(b_2 + a_2 X)W}d X, \nonumber \\
I_2(X) &= \int_\xi^X  \dfrac{y_2 B_{j-1}(X,\tau_{j-1})}{(b_2 + a_2 X)W}d X
\nonumber \\
&= \int_\xi^X \dfrac{y_2 \BP}{(b_2 + a_2 X)W}d X. \nonumber
\end{align*}
Suppose we compute these integrals on the interval $[X_i, X_{i+1}]$, i.e. $X \in [X_i,X_{i+1}]$. As the lower limit of integration $\xi$ it is convenient to choose $\xi = X_i$. Then the coefficient $a_2, b_2$ are constant on this interval, and so are $a, \alpha, \beta$. The homogeneous solutions $y_1, y_2$ should be chosen according to the analysis of Section~\ref{DupireSec}.

\paragraph{\bf $v_{j,i}^1 < 0$} According to Table~\ref{TabWhole}, for negative $v_{j,i}^1$ we have
\begin{align*}
y_1 &= z e^{X/2} U(a+1,2,z), \qquad y_2 = z e^{X/2} M(1+a,2,z), \\
W &= -\dfrac{e^{\mu} }{\Gamma(a_i+1)}, \qquad z = \mu - X. \nonumber
\end{align*}
Therefore,
\begin{align*}
I_2 &= - \Gamma(a+1) e^{-\mu} \int \dfrac{e^{X/2} z M(1+a,2,z)}{b_2 + a_2 X} \BP d X \nonumber \\
&= \dfrac{\Gamma(a+1)}{a_2}\left[\amt J_0 + \bpt e^{\mu} J_1 + e^{\mu} \bpo J_2\right], \\
J_0 &= \int M(1+a,2,z) d z, \quad J_1 = \int e^{-z} M(1+a,2,z) d z, \nonumber \\
J_2 &= \int (\mu - z) e^{-z} M(1+a,2,z) d z = \mu J_1 - J_3, \nonumber \\
J_3 &= \int z e^{-z} M(1+a,2,z) d z. \nonumber
\end{align*}
From \cite{kummerInt1970} after some transformations we obtain
\begin{align*}
J_1 &= \int e^{-z} M(1+a,2,z) d z = \dfrac{1}{a} e^{-z} M(1+a,1,z),  \\
J_0 &= \int  M(1+a,2,z) d z = \dfrac{1}{a} M(a,1,z), \nonumber \\
J_3 &= \int  z e^{-z} M(1+a,2,z) d z = \dfrac{1}{2} z^2 e^{-z} M(a+2,3,z). \nonumber
\end{align*}
Similarly,
\begin{align*}
I_1 &= - \Gamma(a+1) e^{-\mu} \int \dfrac{e^{X/2}z U(1+a,2,z)}{b_2 + a_2 X}
\BP d X  \nonumber \\
& = \dfrac{\Gamma(a+1)}{a_2} e^{-\mu}
\left[ \amt {\mathcal J}_0 + \bpt e^{\mu} {\mathcal J}_1 +  \bpo e^{\mu} {\mathcal J}_2 \right], \\
{\mathcal J}_0 &= \int U(1+a,2,z) d z, \qquad
{\mathcal J}_1 = \int e^{-z} U(1+a,2,z) d z, \nonumber \\
{\mathcal J}_2 &= \int X e^{-z} U(1+a,2,z) d z = \mu {\mathcal J}_1 - {\mathcal J}_3,   \quad
{\mathcal J}_3 = \int z e^{-z} U(1+a,2,z) d z. \nonumber
 \end{align*}
Again, from \cite{kummerInt1970} we can obtain
\begin{align*}
{\mathcal J}_0 &= \int U(a+1,2,z) d z = -\dfrac{1}{a} U(a,1,z), \\
{\mathcal J}_1 &= \int e^{-z} U(a+1,2,z) d z = -e^{-z} U(a,1,z), \\
{\mathcal J}_3 &=\int z e^{-z} U(a,2,z) d z = -z^2 e^{-z} U(a+2,3,z).
\end{align*}

\paragraph{\\ \bf $v_{j,i}^1 > 0$} According to Table~\ref{TabWhole}, for positive $v_{j,i}^1$ we have
\begin{align*}
y_1 &= z e^{-X/2} U(a+1,2,z), \qquad  y_2 = z e^{-X/2} M(1+a,2,z), \\
W &= -\dfrac{e^{\mu} }{\Gamma(a+1)}, \qquad z = \mu + X. \nonumber
\end{align*}
Hence
\begin{align*}
I_2 &= - \Gamma(a+1) e^{-\mu}\int \dfrac{z e^{-X/2} M(1+a,2,z) \BP}{b_2 + a_2 X} d X  \nonumber \\
&= -\dfrac{\Gamma(a+1)}{a_2} e^{-\mu}
\left[ \amt e^{\mu}{\mathcal I}_0 + \bpt {\mathcal I}_1 +  \bpo {\mathcal I}_2 \right], \\
{\mathcal I}_0 &= \int e^{-z} M(1+a,2,z) dz = J_1, \quad
{\mathcal I}_1 = \int M(1+a,2,z) dz = J_0, \nonumber \\
{\mathcal I}_2 &= \int (z-\mu) M(1+a,2,z) d z = {\mathcal I}_3 - \mu J_0, \nonumber \\
{\mathcal I}_3 &= \int z M(1+a,2,z) d z =
\dfrac{z-1}{a}M(a,1,z) + \dfrac{1}{a} M(a-1,1,z). \nonumber
\end{align*}
Similarly,
\begin{align}
I_1 &= - \Gamma(a+1) e^{-\mu} \int \dfrac{z e^{-X/2} U(1+a,2,z) \BP}{b_2 + a_2 X} d X  \nonumber \\
&= -\dfrac{\Gamma(a+1)}{a_2} e^{-\mu} \left[ e^{\mu}\amt {\mathcal P}_2 +
\bpt {\mathcal P}_0 +  \bpo ({\mathcal P}_3 - \mu {\mathcal P}_0) \right], \nonumber \\
{\mathcal P}_0 &= \int U(1+a,2,z) dz = {\mathcal J}_0, \qquad
{\mathcal P}_2 = \int e^{-z} U(1+a,2,z) dz = {\mathcal J}_1, \nonumber \\
{\mathcal P}_3 &= \int z U(1+a,2,z) dz =
-\dfrac{z}{a}\left(U(a,1,z) + \frac{1}{a-1} U(a,2,z) \right). \nonumber
\end{align}

\end{appendices}

\end{document}